\documentclass[11pt]{article}
\usepackage[margin=1in]{geometry} 
\usepackage{amssymb,amsfonts,amsmath,bbm,mathrsfs,stmaryrd}
\usepackage{xcolor}
\usepackage{url}

\usepackage{enumerate}

\usepackage[colorlinks,
             linkcolor=black!75!red,
             citecolor=blue,
             pdftitle={},
             pdfproducer={pdfLaTeX},
             pdfpagemode=None,
             bookmarksopen=true
             bookmarksnumbered=true]{hyperref}

\usepackage{tikz}
\usetikzlibrary{arrows,calc,decorations.pathreplacing,decorations.markings,intersections,shapes.geometric,through,fit,shapes.symbols,positioning,decorations.pathmorphing}

\usepackage{braket}

\numberwithin{equation}{section}

\usepackage[amsmath,thmmarks,hyperref]{ntheorem}
\usepackage{cleveref}

\creflabelformat{enumi}{#2(#1)#3}

\crefname{section}{Section}{Sections}
\crefformat{section}{#2Section~#1#3} 
\Crefformat{section}{#2Section~#1#3} 

\crefname{subsection}{\S}{\S\S}
\AtBeginDocument{%
  \crefformat{subsection}{#2\S#1#3}%
  \Crefformat{subsection}{#2\S#1#3}%
}

\theoremstyle{plain}

\newtheorem{lemma}{Lemma}[section]
\newtheorem{proposition}[lemma]{Proposition}
\newtheorem{corollary}[lemma]{Corollary}
\newtheorem{theorem}[lemma]{Theorem}
\newtheorem{conjecture}[lemma]{Conjecture}

\theoremstyle{nonumberplain}

\theoremstyle{plain}
\theorembodyfont{\upshape}
\theoremsymbol{\ensuremath{\blacklozenge}}

\newtheorem{definition}[lemma]{Definition}

\newtheorem{remark}[lemma]{Remark}

\crefname{definition}{definition}{definitions}
\crefformat{definition}{#2definition~#1#3} 
\Crefformat{definition}{#2Definition~#1#3} 

\crefname{ex}{example}{examples}
\crefformat{example}{#2example~#1#3} 
\Crefformat{example}{#2Example~#1#3} 

\crefname{remark}{remark}{remarks}
\crefformat{remark}{#2remark~#1#3} 
\Crefformat{remark}{#2Remark~#1#3} 

\crefname{convention}{convention}{conventions}
\crefformat{convention}{#2convention~#1#3} 
\Crefformat{convention}{#2Convention~#1#3}

\crefname{lemma}{lemma}{lemmas}
\crefformat{lemma}{#2lemma~#1#3} 
\Crefformat{lemma}{#2Lemma~#1#3} 

\crefname{proposition}{proposition}{propositions}
\crefformat{proposition}{#2proposition~#1#3} 
\Crefformat{proposition}{#2Proposition~#1#3} 

\crefname{corollary}{corollary}{corollaries}
\crefformat{corollary}{#2corollary~#1#3} 
\Crefformat{corollary}{#2Corollary~#1#3} 

\crefname{theorem}{theorem}{theorems}
\crefformat{theorem}{#2theorem~#1#3} 
\Crefformat{theorem}{#2Theorem~#1#3} 

\crefname{assumption}{assumption}{Assumptions}
\crefformat{assumption}{#2assumption~#1#3} 
\Crefformat{assumption}{#2Assumption~#1#3} 

\crefname{equation}{}{}
\crefformat{equation}{(#2#1#3)} 
\Crefformat{equation}{(#2#1#3)}

\theoremstyle{nonumberplain}
\theoremsymbol{\ensuremath{\blacksquare}}

\newtheorem{proof}{Proof}
\newcommand\pf[1]{\newtheorem{#1}{Proof of \Cref{#1}}}

\newcommand\bA{{\mathbb A}}

\newcommand\bC{{\mathbb C}}

\newcommand\bP{{\mathbb P}}
\newcommand\bQ{{\mathbb Q}}

\newcommand\bZ{{\mathbb Z}}

\newcommand\cA{{\mathcal A}}

\newcommand\cC{{\mathcal C}}

\newcommand\cS{{\mathcal S}}

\newcommand{\qedhere}{\mbox{}\hfill\ensuremath{\blacksquare}}

\title{Symbolic dynamics and rotation symmetric Boolean functions}
\author{Alexandru Chirvasitu and Thomas Cusick}

\begin{document}

\date{}

\newcommand{\Addresses}{{
  \bigskip
  \footnotesize
  
  \textsc{Department of Mathematics, University at Buffalo, Buffalo,
    NY 14260-2900, USA}\par\nopagebreak \textit{E-mail address}:
  \texttt{achirvas@buffalo.edu}

  \medskip
  
  \textsc{Department of Mathematics, University at Buffalo, Buffalo,
    NY 14260-2900, USA}\par\nopagebreak \textit{E-mail address}:
  \texttt{cusick@buffalo.edu}

}}

\maketitle

\begin{abstract}
  We identify the weights $wt(f_n)$ of a family $\{f_n\}$ of rotation symmetric Boolean functions with the cardinalities of the sets of $n$-periodic points of a finite-type shift, recovering the second author's result that said weights satisfy a linear recurrence. Similarly, the weights of idempotent functions $f_n$ defined on finite fields can be recovered as the cardinalities of curves over those fields and hence satisfy a linear recurrence as a consequence of the rationality of curves' zeta functions. Weil's Riemann hypothesis for curves then provides additional information about $wt(f_n)$. We apply our results to the case of quadratic functions and considerably extend the results in an earlier paper of ours.
\end{abstract}


\vspace{.5cm}



\section*{Introduction}

Boolean functions $f:\{0,1\}^n\to \{0,1\}$ have long held the interest of the cryptographic community due to their many applications to that field; see e.g. \cite{osw,for-av,ny-perf-zrE,ny-perf-pxx,sz-non,ad-sym} (to give just a few examples that could not possibly do the subject justice) or the discussion and numerous references cited in \cite{cs-bk}.

Among all Boolean functions, the ones with the best cryptographic properties tend to be {\it balanced}, i.e. take the values $0,1$ equally many times. For that reason, one is more generally interested in the {\it weight} $wt(f)$ of a Boolean function, meaning the cardinality of the preimage $f^{-1}(1)$.

We are concerned here with Boolean functions that are {\it rotation symmetric} in the sense of \cite{pq-rot}: those $f$ as above that are invariant under permuting the $n$ variables cyclically. Such a function is expressible as
\begin{equation}\label{eq:38}
  f_n(x_0,\cdots,x_{n-1})=\sum_{i\;\mathrm{mod}\; n}x_i x_{i+a_1}\cdots x_{i+a_{d-1}},\ x_i\in \{0,1\}. 
\end{equation}
In fact, having fixed the $a_j$, the formula \Cref{eq:38} gives rise to a {\it family} of rotation symmetric functions $f_n$ in $n$ variables respectively (see \Cref{subse.bool} below). The starting point for the current paper is the phenomenon constituting the main theorem of \cite{csk-rec2} (which in turn builds on earlier work in the same direction \cite{cs-fast,bc-rec,bcp}) whereby the weights $wt(f_n)$ attached to a family of rotation symmetric Boolean functions satisfy a linear recurrence of the form
\begin{equation*}
  wt(f_{n+N}) = a_{N-1} wt(f_{n+N-1})+\cdots+a_1 wt(f_n),\ \forall \text{ sufficiently large } n. 
\end{equation*}

The motivation here was a desire to understand that recurrence phenomenon in light of other analogous results in the literature to the effect that sequences tracking the sizes of various meaningful sets are linearly recurrent. The general paradigm is that said sequences $N_n$ are collected into a single mathematical object
\begin{equation*}
\zeta(s):=\exp\left(\sum_{n\ge 1} \frac{N_n}n s^n\right).  
\end{equation*}
called a {\it zeta function} and the desired recurrence follows from the rationality of that power series since said rationality, in fact, will say even more:
\begin{equation*}
  N_n = \sum_{i}\alpha_i^n - \sum_j \beta_j^n
\end{equation*}
for algebraic integers $\alpha_i$ and $\beta_j$ known as the {\it characteristic values} of the zeta function.

We consider two instances of this setup, both shedding light on Boolean functions in slightly different ways:
\begin{itemize}
\item zeta functions of dynamical systems \cite{bl}, where $N_n$ is the number of $n$-periodic points under the iterations of a continuous self-map of a compact space, and
\item zeta functions of algebraic varieties \cite{dwork}, with $N_n$ being the number of points of a fixed algebraic variety over the field $GF(q^n)$ with $q^n$ elements for a fixed prime power $q$.  
\end{itemize}

\Cref{se.prel} gathers the needed background material on Boolean functions, symbolic dynamics and algebraic geometry. We also describe the irreducible factors of polynomials of the form $x^{2t}-2^t$ (\Cref{pr.theta-irred}) for later use in \Cref{se.quad}. 

In \Cref{se.rot} we recast rotation symmetric Boolean functions as particular instances of well behaved dynamical systems known as {\it finite-type shifts}: closed subsets of a Cartesian power $\cA^{\bZ}$ of a finite alphabet, invariant under the leftward shift of bi-infinite sequences. These are well studied objects with a rich theory, and in particular it is a fact that their zeta functions are rational. Our main result in that section (\Cref{th.is-fin}) can be paraphrased as

\begin{theorem}\label{th.int-rot}
  For every family of rotation symmetric Boolean functions $f_n$ as in \Cref{eq:38} there is a finite-type shift with $2^{n+1}-2wt(f_n)$ $n$-periodic points for each $n$.

  In particular, $wt(f_n)$ satisfies a linear recurrence.
  \qedhere  
\end{theorem}

\Cref{se.tr} revolves around close cousins of rotation symmetric Boolean functions, definable in Galois-theoretic terms: having fixed a polynomial $P$ with coefficients in the field $GF(2)$, one can consider the family of functions
\begin{equation}\label{eq:39}
  f_n:GF(2^n)\to GF(2),\ f_n(x) = \mathrm{Tr}(P(x)).
\end{equation}
These are introduced in \cite{cgl} and studied there as Boolean functions (which is what they are, having identified $GF(2^n)$ with $GF(2)^n$). The analogue of \Cref{th.int-rot} in this case is almost immediate (\Cref{cor.nrsol-gf}):

\begin{theorem}\label{th.int-tr}
  For a family \Cref{eq:39} of trace functions there is a plane curve $X$ defined over $GF(2)$ such that $2^{n+1}-2wt(f_n)$ is the number of points of $X$ over $GF(2^n)$.

  In particular, since zeta functions of algebraic varieties are rational \cite{dwork}, $wt(f_n)$ again satisfies a linear recurrence.   
  \qedhere
\end{theorem}

We give more precise information on the weights $wt(f_n)$ and the general shape of the recurrence they satisfy in \Cref{cor.no-beta,cor.g} by computing the genus of the curve $X$ of \Cref{th.int-tr} through successive blowups.

Finally, in \Cref{se.quad} we go back to the quadratic case analyzed closely in \cite{us}. In that setup \Cref{th.quad} gives a close connection between the Boolean and trace sides of the picture. Furthermore, for {\it monomial} (quadratic, rotation symmetric) functions
\begin{equation*}
  f_{n,t}(x_i) = \sum_{i\;\mathrm{mod}\; n}x_i x_{i+t}
\end{equation*}
we show in \Cref{th.char} that the characteristic values resulting as in \Cref{se.rot} from the general theory of finite-type shifts precisely coincide with the eigenvalues (including multiplicities) of the recurrence matrix $R(t)$ for $wt(f_{n,t})$ constructed in \cite{csk-rec2}. This is a curious instance of consilience, given how different the methods of \cite{csk-rec2} and \Cref{se.rot} are. 

We hope that the methods of the present paper will not only provide a conceptual explanation for the weight recurrence phenomena so prevalent in the Boolean function literature, but also highlight connections to different areas (symbolic dynamics, algebraic geometry) by bringing to bear tools specific to those fields.

\subsection*{Acknowledgements}

A.C. was partially supported by NSF grant DMS-1801011.

\section{Preliminaries}\label{se.prel}

Throughout, $GF(q)$ denotes the finite field with $q$ elements. We focus primarily on characteristic-two fields, i.e. $q=2^n$. 

\subsection{Boolean functions}\label{subse.bool}

We will work with functions defined on either
\begin{itemize}
\item tuples of Boolean variables. i.e. elements of $V_n=GF(2)^n$, or
\item single finite fields $GF(2^n)$. 
\end{itemize}

We will see that there are strong analogies between these two setups. Specifically, we will construct (following \cite{cgl}) infinite families of functions $f_n$ of the two types (indexed by the respective $n$).

To that end, consider a finite collection $\cC$ of tuples
\begin{equation}\label{eq:20}
  0<a_1<\cdots<a_{d-1}. 
\end{equation}
of positive integers for various $d$. 

We then write $f_{\cC}$ as a collective label for the functions $f_{\cC,n}$ defined in either of the two following ways (to be distinguished contextually in the sequel):

\begin{definition}\label{def.ctxt}
  In {\it rotation symmetric (or RS) context} $f_{\cC,n}$ is the rotation symmetric Boolean function $f_{\cC,n}:V_n\to GF(2)$ obtained as the sum of the {\it monomial RS (or MRS) functions}
  \begin{equation}\label{eq:21}
    (0,a_1,\cdots,a_{d-1}):=\sum_{i\;\mathrm{mod}\; n}x_i x_{i+a_1}\cdots x_{i+a_{d-1}}
  \end{equation}
  as the tuples \Cref{eq:20} range over $\cC$.

  Similarly, in {\it trace context} $f_{\cC,n}$ is the function $f_{\cC,n}:GF(2^n)\to GF(2)$ obtained as the sum of the {\it monomial trace functions}
  \begin{equation*}
    GF(2^n)\ni x\mapsto \mathrm{Tr}\left(x^{1+2^{a_1}+\cdots+2^{a_{d-1}}}\right)
  \end{equation*}
where once more the tuples \Cref{eq:20} range over $\cC$ and $\mathrm{Tr}$ denotes the trace $\mathrm{Tr}_n:GF(2^n)\to GF(2)$.   
\end{definition}

Having fixed $\cC$, there is a close relationship between $f_{\cC}$ in RS and trace context: if $f_n=f_{\cC,n}$ in RS context then the trace context counterpart $g_n=g_{\cC,n}$ is denoted in \cite[Definition 4.1]{cgl} by $f'_n$ and can be obtained from $f$ by
\begin{equation*}
  GF(2^n)\ni x\mapsto f_n(x,x^2,\cdots,x^{2^{n-1}})\in GF(2).
\end{equation*}

\subsection{Symbolic dynamics}\label{subse.symb}

For background on the topic we refer to \cite[Chapters 1-3,6]{lm}. The central notion is

\begin{definition}\label{def.symb}
  Let $\cA$ be a finite set (the {\it alphabet}) and equip the space $\cA^{\bZ}$ of bi-infinite $\cA$-valued sequences
  \begin{equation*}
    \cdots,\ x_{-1},\ x_0,\ x_1,\ \cdots \in \cA
  \end{equation*}
  with its compact Hausdorff product topology. A {\it shift over $\cA$} is a closed subset $X\subseteq \cA^{\bZ}$ preserved by the shift operator
  \begin{equation*}
    \sigma:\cA^{\bZ}\to \cA^{\bZ} 
  \end{equation*}
  defined by $\sigma({\bf x})_i = x_{i+1}$, where
  \begin{equation*}
    {\bf x} = (\cdots,x_{-1},x_0,x_1,\cdots)\in \cA^{\bZ}. 
  \end{equation*}

  We often write $(X,\sigma)$ for a shift, to indicate that we are equipping $X$ with the restriction of the shift map $\sigma$. 

  A {\it subshift} $(Y,\sigma)\subseteq (X,\sigma)$ is a closed subset $Y\subseteq X$ invariant under $\sigma$.  
\end{definition}
This is equivalent to \cite[Definition 1.2.1]{lm}. One particular class of shifts we will be interested in is described in \cite[Introduction]{bl} or \cite[Definition 2.1.1]{lm}.

Before recalling the definition we introduce the following piece of notation: for a finite word
\begin{equation*}
  w\in \cA^*:=\text{possibly-empty words with letters in }\cA
\end{equation*}
we write $X_w$ for the set of elements in $\cA^{\bZ}$ that do not contain $w$ as a subword. More generally, for a set $\cS$ of words we write
\begin{equation*}
  X_{\cS}:=\bigcap_{w\in \cS}X_w=\text{sequences containing no element of }\cS\text{ as a subword}. 
\end{equation*}

It is clear that $X_{\cS}$ is invariant under $\sigma$ and is thus the underlying space of a shift. With this in hand we have 

\begin{definition}\label{def.fin}
  A shift $(X,\sigma)$ over $\cA$ is {\it of finite type} if there is a finite set $\cS$ such that $X=X_{\cS}$. 
\end{definition}
In other words, the finite-type shifts are those describable by requiring that the sequences in question avoid finitely many patterns (or words) over $\cA$.

We also need the following concept (see \cite[Introduction]{bl} or \cite[Definition 6.4.1]{lm}).

\begin{definition}\label{def.zeta}
  Let $(X,\sigma)$ be a shift over the alphabet $\cA$. For each $n\ge 1$ denote by
  \begin{equation*}
    N_n=N_n(X,\sigma)
  \end{equation*}
  the number of elements of $X$ left invariant by $\sigma^{n}$ (i.e. the sequences in $X$ that are $n$-periodic).

  The {\it zeta function} of $(X,\sigma)$ is
  \begin{equation*}
    \zeta(s)=\zeta_{X,\sigma}(s):=\exp\left(\sum_{n\ge 1} \frac{N_n}n s^n\right).
  \end{equation*}  
\end{definition}

One of the important results on zeta functions is \cite[Theorem 1]{bl} (see also \cite[Theorem 6.4.6]{lm}):

\begin{theorem}\label{th.is-rat}
  The zeta function of a finite-type shift is of the form
  \begin{equation*}
    \zeta(s)=\frac 1{\det(1-sA)}
  \end{equation*}
  for some square integer-entry matrix $A$. 
\end{theorem}

\subsection{The Weil conjectures}\label{subse.weil}

A good introduction for this is \cite[Appendix C]{hart}.

Let $X$ be an algebraic variety (typically affine or projective) defined over a finite field $F=GF(q)$ for some prime power $q$. We write $N_n=N_n(X)$ for the number of points of $X$ defined over $GF(q^n)$. Recall (\cite[Appendix C.1]{hart}):

\begin{definition}\label{def.weil-zeta}
  The {\it zeta function of $X$} is
  \begin{equation*}
    \zeta(s)=\zeta_{X}(s):=\exp\left(\sum_{n\ge 1} \frac{N_n}n s^n\right).
  \end{equation*}
\end{definition}

Note the analogy to \Cref{def.zeta}. The {\it Weil conjectures} are a series of statements regarding $\zeta_X(s)$ for {\it smooth} projective varieties $X$ (which thus provide information about the numbers $N_n(X)$). The `conjecture' moniker is preserved for historical reasons: posed in \cite{we-cj} and resolved for curves in \cite{we-bk}, the most difficult of the statements was settled completely in \cite{de1}, so the ``conjectures'' are, in fact, theorems. We refer to \cite[Appendix C.2]{hart} for a more complete historical account.

Since we are concerned primarily with possibly-singular curves $X$, we phrase the results in the more complete form covered in \cite{ap}. Moreover, we focus on the numbers $N_n(X)$ themselves (rather than the zeta function). With this in mind, the relevant statement is \cite[Corollary 2.4]{ap}:

\begin{theorem}\label{th.weil-cv}
  Let $X$ be a projective curve over a finite field $GF(q)$, $\widetilde X\to X$ the normalization of $X$, $g$ the genus of $\widetilde X$ and $\Delta$ the number
  \begin{equation*}
    \left|\widetilde{X}(\overline{GF(q)})-{X}(\overline{GF(q)})\right|.
  \end{equation*}
  Then, there are Galois-invariant multisets of algebraic integers
  \begin{itemize}
  \item $\alpha_i$, $1\le i\le 2g$ with $|\alpha_i|=\sqrt q$;
  \item $\beta_j$, $1\le j\le \Delta$ with $|\beta_j|=1$   
  \end{itemize}
  such that
  \begin{equation}\label{eq:12}
    N_n(X) = q^n+1 - \sum_{i=1}^{2g} \alpha_i^n - \sum_{j=1}^{\Delta} \beta_j^n.
  \end{equation}  
\end{theorem}

\subsection{A remark on scaled roots of unity}\label{subse.rts}

In the discussion below we will need to analyze the spectrum of a unitary matrix with minimal polynomial $x^{2t}-2^t$ for $t\ge 2$. To that end, we have to understand the factorization of that polynomial over the integers. 

It will be convenient to work with the following polynomials: for a positive integer $d$, $\Theta_d(x)$ is obtained from the $d^{th}$ cyclotomic polynomial $\Phi_d$ by
\begin{itemize}
\item substituting $x^2$ for $x$: $\Phi_d(x)\mapsto \Phi_d(x^2)$;
\item scaling all of the resulting roots by $\sqrt 2$, i.e. applying the transformation
  \begin{equation*}
    P(x)\mapsto 2^{\frac{\deg P}2}P\left(\frac x{\sqrt 2}\right)
  \end{equation*}
  to $P(x)=\Phi_d(x^2)$. 
\end{itemize}
More generally, we denote the procedure applied here to $\Phi_d$ (i.e. the two steps above, in succession) by $\alpha$. In other words,
\begin{equation}\label{eq:36}
  (\alpha P)(x) = 2^{\deg P}P\left(\frac{x^2}{2}\right)
\end{equation}
and $\alpha\Phi_d=\Theta_d$. 

Since $x^{2t}-2^t$ is nothing but $\alpha(x^t-1)$, it decomposes as
\begin{equation*}
  x^{2t}-2^t = \prod_{d|t}\Theta_d(x).
\end{equation*}
This makes the following result relevant.

\begin{proposition}\label{pr.theta-irred}
  The polynomial $\Theta_d$ is irreducible except when the exact power of $2$ dividing $d$ is $4$, in which case its irreducible factor decomposition is 
  \begin{equation*}
    \Theta_d(x)=P(x)P(-x)
  \end{equation*}
  for some irreducible polynomial $P$. 
\end{proposition}
\begin{proof}
Let $\Delta_d$ be the set of primitive $d^{th}$ roots of unity and $\Delta_d^{-2}$ its preimage through squaring (i.e. $\Delta_d^{-2}$ is the set of roots of $\Phi_d(x^2)$). Let also $G$ be the absolute Galois group $\mathrm{Gal}(\overline{\bQ}/\bQ)$. We have to argue that $\sqrt 2 \Delta_d^{-2}$
\begin{itemize}
\item breaks up into two $G$-orbits when $d=4(2e+1)$;
\item is a single $G$-orbit otherwise.  
\end{itemize}
The situation is qualitatively different depending on the parity of $d$:

Case 1: odd $d$. We then have the following disjoint unions
\begin{equation*}
  \Delta_d^{-2} = \Delta_d\sqcup \Delta_{2d} = \Delta_d\sqcup -\Delta_d
\end{equation*}
and the conclusion follows from the fact that the fields $\bQ(\sqrt 2)$ and $\bQ(\Delta_d\cup -\Delta_d)$ are linearly disjoint and hence the Galois group of their compositum is simply the product of their respective Galois groups. This affords us the choice to send a fixed primitive $d^{th}$ root of unity to any other such root and $\sqrt 2$ to $\pm\sqrt 2$. 

Case 2: even $d$. This time around
\begin{equation*}
  \Delta_d^{-2} = \Delta_{2d}. 
\end{equation*}
We have $\sqrt 2\in \bQ(\Delta_8)$. Since $\bQ(\Delta_e)$ and $\bQ(\Delta_f)$ are linearly disjoint when $e$ and $f$ are coprime (e.g. \cite[\S IV.1, Theorem 2]{lang-bk}), we have $\sqrt 2\not\in \bQ(\Delta_{2d})$ unless $d$ is divisible by $4$ and we can then repeat the argument used in Case 1.

It thus remains to treat the case $4|d$. The linear disjointness of cyclotomic fields generated by roots of unity of coprime orders allows us to restrict our attention to the case when $d=2^u$ for some $u\ge 2$ (and hence the field whose Galois group we are interested in is $\bQ(\Delta_{2^{u+1}})$). 

When $u=2$ (corresponding to the case when the exact power of $2$ dividing the original $d$ was $4$) one checks immediately that the four-element set $\sqrt 2\Delta_8$ decomposes into two Galois orbits, namely
\begin{equation*}
  \{1\pm i\}\text{ and }\{-1\pm i\}. 
\end{equation*}

When $u\ge 3$ (and hence $8$ divides $d=2^u$) the difference from the preceding discussion is that now $\sqrt 2$ is a sum of {\it even} powers of a fixed primitive $d^{th}$ root of unity $\zeta$, so the image of $\sqrt 2$ through a Galois group element $\zeta\mapsto \pm \zeta^b$ depends only on $b$ (and not on the sign). If $\zeta\mapsto \zeta^b$ fixes $\sqrt 2$ then that same Galois group element maps $\sqrt 2 \zeta\mapsto \sqrt 2 \zeta^b$. If, on the other hand, we have
\begin{equation*}
  \zeta\mapsto \zeta^b,\ \sqrt 2\mapsto -\sqrt 2
\end{equation*}
then the {\it other} Galois group element $\zeta\mapsto -\zeta^b$ will map
\begin{equation}\label{eq:37}
  \sqrt 2\zeta\mapsto \sqrt 2\zeta^b. 
\end{equation}
Either way, the two arbitrary elements of $\sqrt 2\Delta_{2d}$ appearing in \Cref{eq:37} are in the same Galois orbit.
\end{proof}

\section{Rotation-symmetric functions as dynamical systems}\label{se.rot}

Let $f=f_{\cC}$ be a family of Boolean RS functions associated to a collection $\cC$ of tuples \Cref{eq:20} as in \Cref{def.ctxt}. As before, we write $f_n$ for $f$ specialized to the $n$-dimensional vector space $V_n\cong GF(2)^n$ over $GF(2)$.

To $f$ we can also associate polynomial functions $P_{f,n}:V_n\to V_n$ defined for $f=(0,a_1,\cdots,a_{d-1})$ by
\begin{equation}\label{eq:5}
V_n\ni (x_0,\cdots,x_{n-1})\mapsto (x_0x_{a_1}\cdots x_{a_{d-1}},\ x_1x_{a_1+1}\cdots x_{a_{d-1}+1},\ \cdots)\in V_n
\end{equation}
and in general by extending this definition additively. If we now denote the coordinate-sum map
\begin{equation*}
  V_n\ni (x_0,\cdots,x_{n-1})\mapsto \sum x_i \in GF(2)
\end{equation*}
by $\mathrm{Tr}=\mathrm{Tr}_n$ then we have
\begin{equation}\label{eq:2}
  f_n = \mathrm{Tr}_n\circ P_{f,n}. 
\end{equation}

We make note of the following elementary linear algebra fact whose proof we omit.

\begin{lemma}\label{le.easy-ht90}
  Let $k$ be a field and $n$ a positive integer. A vector ${\bf x}\in k^n$ has vanishing sum of coordinates if and only if $x=y-\sigma y$ for some $y\in k^n$, where $\sigma$ is the rotation operator on $k^n$ defined by \Cref{eq:1}.
\end{lemma}

\begin{remark}\label{re.ht90}
  \Cref{le.easy-ht90} is an analogue of the celebrated {\it Hilbert theorem 90}, in its additive version: if $K\subset L$ is a Galois extension with cyclic Galois group $\langle\sigma\rangle$ then an element $x\in L$ has vanishing trace if and only if
  \begin{equation*}
    x=y-\sigma y
  \end{equation*}
  for some $y\in L$.
\end{remark}

\begin{lemma}\label{le.eqn}
  For ${\bf x}\in V_n$ we have $f_n({\bf x})=0$ if and only if
  \begin{equation}\label{eq:3}
    P_{f,n}({\bf x}) = {\bf y}-\sigma {\bf y}
  \end{equation}
  for some ${\bf y}\in V_n$. 
\end{lemma}
\begin{proof}
  This is immediate from \Cref{le.easy-ht90,eq:2}.
\end{proof}

\begin{corollary}\label{cor.nrsol}
  The number of zeros of $f_n$ (i.e. $2^n-wt(f_n)$) is half the number of solutions $({\bf x}, {\bf y})\in V_n^2$ to the equation \Cref{eq:3}.
\end{corollary}
\begin{proof}
  This is immediate from \Cref{le.eqn}, since having fixed ${\bf x}$, a solution ${\bf y}$ to \Cref{eq:3} is uniquely determined up to translation by the all-$1$ vector ${\bf 1}=(1,1\cdots)$.
\end{proof}

Consider embeddings $V_n\to V_{dn}$ for all positive integers $d,n$ defined by
\begin{equation*}
  \iota=\iota_{n,dn}:V_n\ni (x_0,\cdots,x_{n-1})\mapsto (x_0,x_1,\cdots,x_{n-1},x_0,x_1,\cdots)\in V_{dn}. 
\end{equation*}
If we equip each $V_n$ with its rotation operator
\begin{equation}\label{eq:1}
  \sigma:(x_0,\cdots,x_{n-1})\mapsto (x_1,\cdots,x_{n-1},x_0)
\end{equation}
then the $\iota$ embeddings intertwine the respective rotations.

The spaces $V_n$ with connecting maps $\iota_{n,dn}$ form a diagram in the category of vector spaces whose colimit (i.e. union) we denote by $V_{\infty}$. The latter is nothing but the space of bi-infinite periodic sequences over $GF(2)$. Furthermore, the rotation operators $\sigma$ on the various $V_n$ lift precisely to the shift on $V_{\infty}$ in the sense of \Cref{subse.symb} (again denoted by $\sigma$).

\begin{remark}\label{re.like-fields}
  The analogy noted in \Cref{re.ht90} extends further: the $\iota:V_n\to V_{dn}$ parallel the inclusions $GF(2^n)\subset GF(2^{dn})$ of finite fields, the shift operator $\sigma$ on $V_{\infty}$ is similar in spirit to the Frobenius automorphism $x\mapsto x^2$ of the algebraic closure $\overline{GF(2)}$, etc.
\end{remark}

The polynomial functions $P_{f,n}:V_n\to V_n$ fit into commutative diagrams
\begin{equation*}
 \begin{tikzpicture}[auto,baseline=(current  bounding  box.center)]
  \path[anchor=base] (0,0) node (1) {$V_n$} +(2,.5) node (2) {$V_{dn}$} +(2,-.5) node (3) {$V_{n}$} +(4,0) node (4) {$V_{dn}$};
  \draw[->] (1) to[bend left=6] node[pos=.5,auto] {$\scriptstyle \iota$} (2);
  \draw[->] (1) to[bend right=6] node[pos=.5,auto,swap] {$\scriptstyle P_{f,n}$} (3);
  \draw[->] (2) to[bend left=6] node[pos=.5,auto] {$\scriptstyle P_{f,dn}$} (4);
  \draw[->] (3) to[bend right=6] node[pos=.5,auto,swap] {$\scriptstyle \iota$} (4);
 \end{tikzpicture}
\end{equation*}
and hence give rise to a map $P_f:V_{\infty}\to V_{\infty}$. This is significant because it will allow us, in a sense, to lift the zero-counting for $f_n$ from $V_n$ to $V_{\infty}$ via \Cref{le.eqn}: while the trace map $\mathrm{Tr}$ does not make sense on $V_{\infty}$, the equation
\begin{equation}\label{eq:4}
  P_f({\bf x}) = {\bf y}-\sigma {\bf y}
\end{equation}
does. In fact, we can do more: the definition of $P_f$ extends in the obvious fashion to a self-map of the entire sequence space $\Sigma:=GF(2)^{\bZ}$ by mimicking the definition in \Cref{eq:5} for monomials and then extending additively, as before. Note that $P_f$ will then be a shift intertwiner, in the sense that 
\begin{equation*}
 \begin{tikzpicture}[auto,baseline=(current  bounding  box.center)]
  \path[anchor=base] (0,0) node (1) {$\Sigma$} +(2,.5) node (2) {$\Sigma$} +(2,-.5) node (3) {$\Sigma$} +(4,0) node (4) {$\Sigma$};
  \draw[->] (1) to[bend left=6] node[pos=.5,auto] {$\scriptstyle \sigma$} (2);
  \draw[->] (1) to[bend right=6] node[pos=.5,auto,swap] {$\scriptstyle P_{f}$} (3);
  \draw[->] (2) to[bend left=6] node[pos=.5,auto] {$\scriptstyle P_{f}$} (4);
  \draw[->] (3) to[bend right=6] node[pos=.5,auto,swap] {$\scriptstyle \sigma$} (4);
 \end{tikzpicture}
\end{equation*}
commutes. 

We are now in a position to associate a shift $(X_f,\sigma)$ to each RS function $f$: define $X_f$ to be the subspace of
\begin{equation*}
(GF(2)\times GF(2))^{\bZ} \cong GF(2)^{\bZ}\times GF(2)^{\bZ} 
\end{equation*}
consisting of the $({\bf x},{\bf y})$ satisfying \Cref{eq:4}. The shift map $\sigma$ on $X_f$ will simply be the restriction of diagonal $(\sigma,\sigma)$ to
\begin{equation*}
  X_f\subset GF(2)^{\bZ}\times GF(2)^{\bZ}. 
\end{equation*}

The following remark captures the relationship between the zeros of $f_n$ and the shift $(X_f,\sigma)$.

\begin{lemma}\label{le.is-shift}
  With the notation above we have
  \begin{equation*}
    2^{n+1}-2wt(f_n) = N_n(X_f,\sigma)
  \end{equation*}
  with $N_n$ denoting number of fixed points of $\sigma^n$, as in \Cref{subse.symb}. 
\end{lemma}
\begin{proof}
  This follows from \Cref{cor.nrsol} after noting that
  \begin{itemize}
  \item $V_{\infty}\subset GF(2)^{\bZ}$ consists precisely of the periodic sequences, so the elements $({\bf x},{\bf y})$ contributing to $N_n$ belong to $V_{\infty}\times V_{\infty}$.
  \item $V_n$ is identifiable with the fixed-point set of $\sigma^n$ in $V^{\infty}$.   
  \end{itemize}
\end{proof}

We can now analyze the shift $(X_f,\sigma)$ for the purpose of extracting interesting properties for the function $n\mapsto wt(f_n)$ via \Cref{le.is-shift}.

\begin{theorem}\label{th.is-fin}
For any RS function $f$ the associated shift $(X_f,\sigma)$ is of finite type.   
\end{theorem}
\begin{proof}
  Recall that by definition,
  \begin{equation*}
    X_f\subset (GF(2)\times GF(2))^{\bZ}
  \end{equation*}
  consists of those pairs of elements ${\bf x}$, $\bf y$ in $GF(2)^{\bZ}$ satisfying \Cref{eq:4}, paraphrased here as
  \begin{equation}\label{eq:6}
    P_f({\bf x})-({\bf y}-\sigma{\bf y})={\bf 0}\in GF(2)^{\bZ}. 
  \end{equation}
  The left hand side of \Cref{eq:6} constitutes a shift-equivariant polynomial map
\begin{equation*}
  (GF(2)\times GF(2))^{\bZ}\ni ({\bf x},{\bf y})\mapsto Q({\bf x},{\bf y})\in GF(2)^{\bZ}, 
\end{equation*}
in the sense that
\begin{itemize}
\item there is some finite interval $I\subset \bZ$ such that
  \begin{equation*}
    Q({\bf x},{\bf y})_0 = \text{polynomial }R(x_i,y_j)\text{ for }i,j\in I
  \end{equation*}
  (that justifies the term ``polynomial'') and
\item $Q(\sigma{\bf x},\sigma{\bf y}) = \sigma Q({\bf x},{\bf y})$
  (i.e. ``shift-equivariant'').
\end{itemize}
In other words, $X_f$ consists precisely of those sequences of elements in $GF(2)\times GF(2)$ which do not contain, as subwords, the finitely many non-solutions to
\begin{equation*}
  R(x_i,y_j)=0,\ i,j\in I. 
\end{equation*}
This makes it clear that the shift is indeed of finite type. 
\end{proof}

In particular, \Cref{th.is-rat} and \Cref{le.is-shift} then proves

\begin{corollary}\label{cor.wtzeta}
  Let $f$ be an RS Boolean function and set
  \begin{equation*}
    N_n=2^{n+1}-2wt(f_n) = 2^n + W_f({\bf 0}).
  \end{equation*}
  Then, we have
  \begin{equation*}
    \exp\left(\sum_{n\ge 1} \frac{N_n}n s^n\right) = \frac 1{\det(1-sA)} 
  \end{equation*}
  for some integer square matrix $A$. 
\end{corollary}

Or again:

\begin{corollary}\label{cor.sum-pows}
  For any RS Boolean function $f$ there are algebraic integers $\alpha_i$, $1\le i\le r$ such that
  \begin{enumerate}[(a)]
  \item the multiset $\{\alpha_i\}$ is closed under Galois conjugation over $\bQ$, and
  \item we have
    \begin{equation*}
      wt(f_n) = 2^n - \frac{\alpha_1^n+\cdots+\alpha_r^n}2.
    \end{equation*}
  \end{enumerate}
\end{corollary}
\begin{proof}
  Let $\alpha_i$, $1\le i\le r$ be the eigenvalues (with multiplicity) of the integer matrix $A$. We have
  \begin{equation*}
    \frac 1{\det(1-sA)} = \exp\left(\sum_{n\ge 1}\frac{\alpha_1^n+\cdots+\alpha_r^n}n s^n\right),
  \end{equation*}
  so
  \begin{equation*}
    2^{n+1} - 2wt(f_n) = N_n = \alpha_1^n+\cdots+\alpha_r^n
  \end{equation*}
  for $N_n$ as in \Cref{cor.wtzeta}. This completes the proof.
\end{proof}

Finally, as an immediate consequence of \Cref{cor.sum-pows} we obtain
\begin{corollary}
  The weights $wt(f_n)$ of an RS Boolean function $f$ satisfy a linear recurrence with integer coefficients. 
\end{corollary}
This provides a new proof for the existence of the linear recurrences, which can be computed using the results in \cite{csk-rec1, csk-rec2}.

\section{Trace representations and the Weil conjectures}\label{se.tr}

We now give a parallel treatment for trace-context functions $f_n=f_{\cC,n}:GF(2^n)\to GF(2)$ as in \Cref{def.ctxt}. One is again interested in the weights $wt(f_n)$, i.e. the cardinalities of the sets $f_n^{-1}(1)\subset GF(2^n)$. 

The analogue of \Cref{le.easy-ht90} in the present setting is precisely the Hilbert theorem 90 recalled in \Cref{re.ht90}:

\begin{lemma}
  Let $n$ be a positive integer. An element $x\in GF(2^n)$ has vanishing trace if and only if $x=y-\sigma y$, where $\sigma$ is the Frobenius automorphism $y\mapsto y^2$ on $GF(2^n)$. 
\end{lemma}

As in \Cref{se.rot}, we introduce the polynomials
\begin{equation*}
  P_{f,n}: GF(2^n)\to GF(2^n)
\end{equation*}
defined for monomials \Cref{eq:20} by
\begin{equation*}
  P_{f,n}(x) = x\cdot x^{2^{a_1}}\cdots x^{2^{a_{d-1}}}
\end{equation*}
and extended additively from this in general. These are restrictions to $GF(2^n)$ of a single polynomial $P_f$ defined on the entire algebraic closure $\overline{GF(2)}$. We now have $f_n=\mathrm{Tr}_n\circ P_{f,n}$ (as in the RS case), hence the following versions of \Cref{le.eqn} and \Cref{cor.nrsol}.

\begin{lemma}\label{le.eqn-gf}
  For $x\in GF(2^n)$ we have $f_n(x)=0$ if and only if
  \begin{equation}\label{eq:8}
    P_{f,n}(x) = y-y^2 
  \end{equation}
  for some $y\in GF(2^n)$. 
\end{lemma}

\begin{corollary}\label{cor.nrsol-gf}
  The number of zeros of $f_n$ (i.e. $2^n-wt(f_n)$) is half the number of solutions $({x}, {y})\in GF(2^n)^2$ to the equation \Cref{eq:8}.  
\end{corollary}

We will now repurpose the notation from \Cref{se.rot}: $X_f$ will denote the affine plane algebraic curve
\begin{equation}\label{eq:9}
  X_f=\{(x,y)\in \overline{GF(2)}^2\ |\ P_f(x) = y-y^2\}. 
\end{equation}

With this notation, \Cref{cor.nrsol-gf} says that we have
\begin{equation}\label{eq:11}
  2^{n+1}-2wt(f_n) = N_n(X_f). 
\end{equation}

We would now like to apply the point count in \Cref{th.weil-cv} to the curve $X_f$ with $q=2$. The only slight obstacle is that theorem applies to {\it projective} curves, whereas $X_f$ is affine. Its closure $X'_f$ in the projective plane $\bP^2$ over the algebraic closure $\overline{GF(2)}$ is given by the {\it homogenization} of the defining equation 
\begin{equation*}
  P_f(x) = y-y^2
\end{equation*}
in \Cref{eq:9}:
\begin{equation}\label{eq:10}
  X'_f=\{[x:y:z]\in \bP^2\ |\ \overline{P}_f(x,z) = yz^{e-1}-y^2z^{e-2}\}
\end{equation}
where
\begin{itemize}
\item $e$ is the largest degree of a monomial in $P_f$, and  
\item $\overline{P}_f$ is the homogeneous degree-$e$ polynomial in $x,z$ obtained by multiplying each monomial of $P_f(x)$ by the appropriate power of $z$.  
\end{itemize}

\begin{remark}\label{re.e}
  $e$ is of the form
  \begin{equation*}
    1+2^{a_0}+\cdots+2^{a_{d-1}} 
  \end{equation*}
for a tuple \Cref{eq:20} and is thus odd and $\ge 3$.   
\end{remark}

Now, note that the original affine curve $X_f$ consists precisely of those points in its projective completion \Cref{eq:10} with $z=0$. Since exactly one of the monomials in $\overline{P}_f(x,z)$ is a power of $x$, we have
\begin{equation*}
  [x:y:z]\in X'_f,\ z=0\Rightarrow x=0\Rightarrow [x:y:z] = [0,1,0]=:p_0. 
\end{equation*}
In other words, the affine curve is missing exactly one point of its completion:
\begin{equation*}
  |X'_f(GF(2^n))| - |X_f(GF(2^n))|=1,\ \forall n\ge 1. 
\end{equation*}
In other words, the version of \Cref{th.weil-cv} applicable to $X_f$ simply omits the `$+1$' summand in that statement:

\begin{theorem}\label{th.alpha-beta}
  Let $f_n$, $n\ge 1$ be a family of trace functions $GF(2^n)\to GF(2)$ attached to a finite set of tuples \Cref{eq:20}. Then, there are Galois-invariant multisets of algebraic integers
  \begin{itemize}
  \item $\alpha_i$, $1\le i\le 2g$ with $|\alpha_i|=\sqrt 2$;
  \item $\beta_j$, $1\le j\le \Delta$ with $|\beta_j|=1$      
  \end{itemize}
  such that
  \begin{equation}\label{eq:17}
    wt(f_n) = 2^{n-1} + \frac{\sum_{i=1}^{2g} \alpha_i^n}2 + \frac{\sum_{j=1}^{\Delta} \beta_j^n}2.
  \end{equation}
\end{theorem}
\begin{proof}
  Simply apply \Cref{th.weil-cv} to the projective curve $X'_f$, omit the `$+1$' term in \Cref{eq:12} as explained above, and use \Cref{eq:11} to identify $N_n(X_f)$ with $2^{n+1}-2wt(f_n)$. The rest is simple arithmetic.
\end{proof}

\Cref{th.weil-cv} makes it clear that the size $\Delta$ of the set of $\beta_j$ depends on ``how singular'' the projective curve in question is. For that reason, it will be of interest to understand the singularities of our curve $X'_f$ defined in \Cref{eq:10}. Writing
\begin{equation*}
  Q(x,y,z)=Q_f(x,y,z) := \overline{P}_f(x,z)-yz^{e-1}+y^2z^{e-2}
\end{equation*}
for the homogeneous degree-$e$ polynomial whose vanishing defines $X'_f$. The singularities of the latter are the points where
\begin{equation*}
  \frac{\partial Q}{\partial x} = \frac{\partial Q}{\partial y} = \frac{\partial Q}{\partial z}=0.
\end{equation*}
The partial derivative $\frac{\partial Q}{\partial y}$ is nothing but $z^{e-1}$ (because we are in characteristic $2$ and hence the derivative of $y\mapsto y^2$ vanishes), so the singular set of $X'_f$ is either empty or precisely 
\begin{equation*}
  \{p_0\}=\{[0:1:0]\} = X'_f\setminus X_f. 
\end{equation*}
As for whether or not $p_0$ is indeed singular, we first observe that the $x$ and $y$ partial derivatives do indeed vanish, leaving the question of whether $\frac{\partial Q}{\partial z}$ does. Recall from \Cref{re.e} that $e$ is odd and hence all powers of $z$ appearing in $\overline{P}_f(x,z)$ are even. It follows that the $z$-partial derivative of $\overline{P}_f(x,z)$ vanishes, so
\begin{equation*}
  \frac{\partial Q}{\partial z}(p_0) = y^2z^{e-3}. 
\end{equation*}
This is zero (and hence the point is singular) when $e>3$ and non-zero when $e=3$. We thus have two possibilities:
\begin{enumerate}[(a)]
\item\label{item:1} $e=3$, in which case $X'_f$ is an elliptic curve;
\item\label{item:2} $e>3$, in which case $X'_f$ is a projective plane curve with a single singularity at $[0:1:0]$.  
\end{enumerate}

We now focus on case \labelcref{item:2}, seeking to determine the discrepancy between $X'=X'_f$ and its desingularization. First, we focus attention on the affine portion $C$ of $X'$ corresponding to $y\ne 0$. Making the variable change
\begin{equation*}
  u=\frac xy,\ v=\frac zy,
\end{equation*}
we can describe $C$ as the curve in the $u,v$ plane defined by the equation
\begin{equation}\label{eq:13}
  \overline{P}_f(u,v)+v^{e-1}+v^{e-2}=0 
\end{equation}
(where we have dropped minus signs, since we are working in characteristic two). We now proceed to resolve the singularity $p_0=(0,0)$ (in $u,v$ coordinates) of $C$ by the procedure described in \cite[Theorem V.3.9 and surrounding discussion]{hart}, of successive blowup. 

The initial blowup of the curve $C\in \bA^2$ (the affine plane) defined by \Cref{eq:13} centered at the singularity $(0,0)$ is achieved as described on \cite[pp.29-30]{hart}:
\begin{itemize}
\item Introduce coordinates $\alpha,\beta$ for the projective line $\bP^1$ over $\overline{GF(2)}$. 
\item Consider the subvariety $V$ of $\bA^2\times \bP^1$ cut out by \Cref{eq:13} and the equation
  \begin{equation*}
    u\beta=v\alpha. 
  \end{equation*}
  $V$ is the union of the distinguished projective line $E:=\{(0,0)\}\times \bP^1$ and the blowup $C_1$ of the original curve $C=C_0$. 
\item As in \cite[Example I.4.9.1]{hart}, we now cover $E$ with the open affine patches $\alpha\ne 0$ and $\beta\ne 0$ and determine the intersection of $C_1$ with each open patch in order to determine the preimage of the singularity $(0,0)$ through the rational map $C_1\to C$.  
\end{itemize}

In this last step, assume first that $\alpha\ne 0$. By rescaling we can thus assume $v=u\beta$. Making this substitution in \Cref{eq:13} we obtain
\begin{equation}\label{eq:14}
  u^eQ(\beta)=u^{e-1}\beta^{e-1}+u^{e-2}\beta^{e-2}
\end{equation}
for some polynomial in $\beta$ with free term $1$, so that $Q=1+R$ with $R(0)=0$. 

When $u=0$ we have $v=0$ as well, and the equations describe $E$. In order to determine its intersection with $C_1$ assume $u\ne 0$ in \Cref{eq:14} and divide through by $u^{e-2}$ to obtain
\begin{equation}\label{eq:15}
 u^2 (1+R(\beta)) = u\beta^{e-1} + \beta^{e-2} = \beta^{e-2}(1+u\beta). 
\end{equation}
The only solution to this equation with $u=0$ is the point $[\alpha:\beta]=[1:0]$ on $E\cong \bP^1$.

A similarly simple calculation shows that $C_1\cap E$ contains {\it no} points in the open patch $\beta\ne 0$. In conclusion, the partial desingularization $C_1\to C_0$ of $(0,0)$ provides a single singular point, obtained as $(0,0)$ on curve defined by \Cref{eq:15} in the $u,\beta$ plane.

Let $A$ be the localization of the ring
\begin{equation*}
  \overline{GF(2)}[u,\beta]/(\text{equation }\Cref{eq:15})
\end{equation*}
at the ideal $(u,\beta)$ and $\widehat{A}$ its completion with respect to its maximal ideal. In other words, $\widehat{A}$ is the formal power series ring
\begin{equation*}
  \overline{GF(2)}[[u,\beta]]
\end{equation*}
modulo the equation \Cref{eq:15}. 

Since we are in characteristic $2$ and $e-2$ is odd, $1+u\beta$ and $1+R(\beta)$ are both (invertible) $(e-2)^{nd}$ powers in $\widehat{A}$. This means that we can make a change of variables
\begin{equation*}
  u\mapsto u,\ \beta\mapsto \gamma=g(\beta)
\end{equation*}
in $\widehat{A}$ so as to transform \Cref{eq:15} into
\begin{equation}\label{eq:16}
  u^2=\gamma^{e-2}.
\end{equation}
In the language of \cite[\S I.5]{hart}, the $(0,0)$ singularity of \Cref{eq:15} is {\it analytically isomorphic} to the $(0,0)$ singularity of \Cref{eq:16}. But the singularities of the form \Cref{eq:16} are analyzed in \cite[Example V.3.9.5]{hart}: they are resolved through a sequence of blowups
\begin{equation}\label{eq:19}
  C_{\frac{e-1}2}\to \cdots\to C_2\to C_1,
\end{equation}
with each $C_i$ having a single singular point.

This analysis will allow us to sharpen \Cref{th.alpha-beta} in two ways. First, since we have just established that the desingularization $\widetilde{X'_f}\to X'_f$ has a {\it unique} point mapping to the singularity of $X'_f$, \Cref{th.weil-cv} says that in fact $\Delta=0$, i.e. there are no $\beta$s in \Cref{eq:17}:

\begin{corollary}\label{cor.no-beta}
  Let $f_n$, $n\ge 1$ be a family of trace functions $GF(2^n)\to GF(2)$ attached to a finite set of tuples \Cref{eq:20}. Then, there is a Galois-invariant multiset of algebraic integers
  \begin{equation*}
    \alpha_i,\ 1\le i\le 2g \text{ with } |\alpha_i|=\sqrt 2
  \end{equation*}
  such that
  \begin{equation}\label{eq:18}
    wt(f_n) = 2^{n-1} + \frac{\sum_{i=1}^{2g} \alpha_i^n}2.
  \end{equation}  
\end{corollary}

Secondly, we can determine the genus $g$ of the desingularization $\widetilde{X'_f}\to X'_f$ (i.e. the $g$ appearing in \Cref{eq:18}). This will require stepping through the desingularization procedure by successive blowup sketched above, using the numerical information provided by \cite[Example V.3.9.2]{hart}.

The latter says that the genus $g$ of the smooth curve $\widetilde{X'_f}$ is obtained from the arithmetic genus $p_a(X'_f)$ by subtracting
\begin{equation*}
  \sum_p\frac{r_p(r_p-1)}2
\end{equation*}
for all singular points appearing during the successive blowups, where $r_p$ is the {\it multiplicity} of the singular point $p$.

We now assemble the ingredients:
\begin{itemize}
\item The arithmetic genus $p_a(X'_f)$ is
  \begin{equation*}
    \frac{(e-1)(e-2)}2,
  \end{equation*}
  since $e$ is the degree of the plane curve $X'_f\subset \bP^2$ (\cite[Exercise I.7.2]{hart}).
\item The multiplicity of the singularity $(0,0)$ on a plane curve is the smallest degree appearing in an expansion of its defining equation. It is thus $e-2$ for the initial singularity \Cref{eq:13} and $2$ for each of the subsequent $\frac{e-3}2$ desingularization steps in \Cref{eq:19}.  
\item In conclusion, the genus $g$ is
  \begin{equation*}
    \frac{(e-1)(e-2)}2 - \frac{(e-2)(e-3)}2 - \frac{e-3}2 = \frac{e-1}2.
  \end{equation*}
\end{itemize}

In short:

\begin{corollary}\label{cor.g}
  The number $2g$ of summands in \Cref{eq:18} is $e-1$, where $e$ is the degree
  \begin{equation*}
    1+2^{a_0}+\cdots+2^{a_{d-1}}
  \end{equation*}
  of $P_f$. 
\end{corollary}

\section{Quadratic functions}\label{se.quad}

By `quadratic' we mean functions $f_{\cC,n}$ (in either the trace or RS setup) associated to collections $\cC$ of tuples \Cref{eq:20} with $d=2$. In that case we refer to $\cC$ itself (or to its members) as being quadratic. The functions $f_{\cC,n}$ for quadratic $\cC$ form the focus of the present section.

First, it is well known that quadratic Boolean functions are {\it plateaued}: the weight of a quadratic Boolean function $f_n$ either vanishes or is of the form $2^{n-1}\pm 2^{\frac{n+v}2-1}$ for some integer $v$ of the same parity as $n$. the same applies in trace context, since trace functions as in \Cref{def.ctxt} can always be regarded as quadratic Boolean functions after choosing an appropriate basis for $GF(2^n)$ (see \cite[Remark 3.2]{cgl}).

Now let $\cC$ be a finite collection of quadratic tuples \Cref{eq:20} and $f_{\cC}$, $g_{\cC}$ the RS and trace function families attached to $\cC$ respectively. With this in place, \cite[Theorem 5.1]{cgl} implies that $f_n$ and $g_n$ are plateaued for the same parameter $v=v(n)$.

We can in fact say more \cite[Theorem 2.1]{us}:

\begin{theorem}\label{th.rs-tr}
  Let $\cC$ be a quadratic family of tuples \Cref{eq:20} and $f_{\cC,n}$, $g_{\cC,n}$ the RS and respectively trace functions attached to it. Then, we either have $wt(f_{\cC,n})=0=wt(g_{\cC,n})$ or
  \begin{align*}
    wt(f_{\cC,n}) &= 2^{n-1}\pm 2^{\frac{n+v}2 -1}\\
    wt(g_{\cC,n}) &= 2^{n-1}\pm 2^{\frac{n+v}2 -1}
  \end{align*}
  for the same $v=v(n)$ (but perhaps not the same sign). 
\end{theorem}

We have
\begin{equation*}
  v(n) = \deg\mathrm{gcd}(x^n-1,A_n(x))
\end{equation*}
where
\begin{equation*}
  A_n(x)=\sum_{(0,t)\in \cC}(x^t+x^{n-t} )
\end{equation*}
and the greatest common divisor is taken in the polynomial ring $GF(2)[x]$. It follows that $v(n)$ is periodic in $n$, and reaches its maximal value once per period: precisely when $n$ is divisible by
\begin{equation}\label{eq:23}
  N=N_{\cC} = \min\{n\text{ such that } A(x)\ |\ x^n-1 \}
\end{equation}
where
\begin{equation*}
  A(x) = \sum_{(0,t)\in \cC}(x^t+x^{-t} )
\end{equation*}
and divisibility takes place in the Laurent polynomial ring $GF(2)[x^{\pm 1}]$. For all of this we refer to \cite[Theorem 5.2]{us}.

All of this additional information available in the quadratic case allows us to recast \Cref{cor.sum-pows,cor.no-beta} as follows.

\begin{theorem}\label{th.quad}
  Let $\cC$ be a finite collection of quadratic tuples \Cref{eq:20} and $f_{\cC,n}$, $g_{\cC,n}$ the RS and trace functions associated to $\cC$ respectively. Then, there are Galois-invariant multisets
  \begin{equation*}
    |\alpha_i|=\sqrt 2 = |\gamma_j|,\ 1\le i,j\le max_n 2^{\frac {v(n)}2}
  \end{equation*}
such that  
  \begin{equation*}
    wt(f_n) = 2^{n-1}-\frac{\sum\alpha_i^n}2
  \end{equation*}
  and
  \begin{equation*}
    wt(g_n) = 2^{n-1}+\frac{\sum\gamma_i^n}2. 
  \end{equation*}
  Furthermore, the degree of the group generated by the roots of unity $\frac{\alpha_i}{\sqrt 2}$ (resp. $\frac{\gamma_j}{\sqrt 2}$) is either the period $N=N_{\cC}$ of $v(n)$ or $2N$. 
\end{theorem}
\begin{proof}
  To fix ideas, we focus on the trace functions $g_n=g_{\cC,n}$. The RS half of the statement will follow from this and \cite[Theorem 5.1]{cgl} (which says that  $f_{\cC,n}$ and  $g_{\cC,n}$ have the same nonlinearity) or \Cref{th.rs-tr}.

  On the one hand, we know that
  \begin{equation}\label{eq:24}
    wt(g_n) = 2^{n-1}\pm 2^{\frac {n+v}2-1}\text{ or }0.
  \end{equation}
  On the other hand, by \Cref{cor.no-beta} we have
  \begin{equation}\label{eq:22}
    wt(g_n) = 2^{n-1}+\frac{\sum \gamma_j^n}2
  \end{equation}
  for a Galois-invariant multiset of algebraic integers $\gamma_j$ of absolute value $\sqrt 2$. We write
  \begin{equation*}
    \chi_j = \frac{\gamma_j}{\sqrt 2}
  \end{equation*}
  for the phases of $\gamma_j$. 

  The fact that the number of $\gamma_j$ is $\max 2^{\frac v2}$ follows by comparing \Cref{eq:24,eq:22}: the former says that the maximal absolute value of
  \begin{equation}\label{eq:7}
    \frac{wt(g_n)-2^{n-1}}{2^{\frac n2 -1}}
  \end{equation}
is $\max 2^{\frac v2}$ while the latter shows it is the size of the multiset $(\gamma_j)_j$.

  We now observe that
  \begin{itemize}
  \item By \Cref{eq:22}, the \Cref{eq:7} is maximized precisely when all $\gamma_j^n$ are positive, i.e. $n$ is divisible by
    \begin{equation*}
      \mathrm{ord}(\gamma_j,j):=\left|\text{group generated by the phases of }\gamma_j\right|. 
    \end{equation*}
  \item On the other hand, \Cref{eq:24} shows that \Cref{eq:7} is maximized in absolute value if and only if $wt(g_n)\ne 0$ and $v(n)$ is maximal, i.e. $n$ is divisible by the period $N=N_{\cC}$ defined in \Cref{eq:23}.
  \item \Cref{eq:7} can be maximized in absolute value but negative only there is some $n$ so that $\chi_j^n=-1$ for all $j$. 

  \end{itemize}
  We now consider several possibilities.  
  \begin{enumerate}[(a)]
  \item\label{item:3} {\bf $wt(g_N)\ne 2^{N-1}$ and \Cref{eq:7} is positive for $n=N$. } In this case the remarks above show that \Cref{eq:7} achieves its maximal value at $n=N$ and hence $\mathrm{ord}(\gamma_j,j)=N$. 
  \item\label{item:4} {\bf $wt(g_N)\ne 2^{N-1}$ and \Cref{eq:7} is negative for $n=N$. } This means that \Cref{eq:7} is maximal in absolute value but negative at $n=N$ and hence $\chi_j^N=-1$ for all $j$. But this then implies that $\mathrm{ord}(\gamma_j,j)=2N$. 
  \item {\bf $wt(g_N)=2^{N-1}$. } We know from \cite[Corollary 5.19]{us} that $g_{2N}$ is {\it not} balanced, i.e. $wt(g_{2N})\ne 2^{2N-1}$. We can now reiterate the arguments in cases \labelcref{item:3} and \labelcref{item:4} with $2N$ in place of $N$ to conclude that
    \begin{equation*}
      \mathrm{ord}(\gamma_j,j) = 2N\text{ or }4N.
    \end{equation*}
  \end{enumerate}
  This finishes the proof of the theorem.
\end{proof}

The following result is an offshoot of the proof of \Cref{th.quad}. 

\begin{theorem}\label{cor.rec-ord}
  Under the hypotheses of \Cref{th.quad} the weights $wt(f_n)$ and $wt(g_n)$ satisfy linear recurrences of orders $\le 2N+1$, where $N$ is the period of the sequence $v(n)$ of plateau parameters.
\end{theorem}
\begin{proof}
  In the language of \Cref{th.quad}, consider the various possible values for $\mathrm{ord}(\gamma_j,j)$:

  If it is $\le 2N$ then we are done. Indeed, by \Cref{eq:22} the weight $w(g_n)$ is a linear combination of $n^{th}$ powers of algebraic integers satisfying the polynomial equation
  \begin{equation*}
    (x-2)(x^{2N}-2^N)=0.
  \end{equation*}

  On the other hand, it follows from the proof of \Cref{th.quad} that the case
  \begin{equation*}
    \mathrm{ord}(\gamma_j,j)=4N
  \end{equation*}
  occurs only when $\gamma_j^{2N}=-2^N$ for all $j$. This means we have recursion polynomial 
  \begin{equation*}
    (x-2)(x^{2N}+2^N)=0.
  \end{equation*}
  instead.

  Either way, the minimal recursion polynomial will have degree $\le 2N$.
\end{proof}

It turns out that, for the quadratic MRS function $(0,t)_n=h_{t,n},$ say, in the RS context, the algebraic integers $\alpha_1, \ldots, \alpha_{2^t}$ from \Cref{cor.sum-pows} and \Cref{th.alpha-beta}
can be taken to be the roots of the characteristic polynomial of the square matrix $R(t)$ associated with $h_{t,n}$ in \cite[Section 3]{us}. The next theorem proves this. Note that this count of the numbers   $\alpha_i$  agrees with the count given in \Cref{eq:18} and \Cref{cor.g}.  The matrix $R(t)$ has $2^t$ rows and is given explicitly in \cite[Theorem 3.1]{us}.

To describe the matrix we need the cyclic permutation $\mu$ which acts on vectors $(b_1, b_2, \ldots, b_k)$ of any length $k$ by placing the last entry to the front, e.g.
  \begin{equation*}
    \mu^2((0,1,0,1,0,0))=(0,0,0,1,0,1).
  \end{equation*}
  We use $0_j$ to stand for a string of $j$ consecutive entries equal to $0$ and similarly for $1_j.$ Now $R(t)$ is the square matrix whose rows are the $2^{t-1}$ pairs
$$\mu^i((1,0_{2^{t-1}-1},1,0_{2^{t-1}-1})),~\mu^i((1,0_{2^{t-1}-1},-1,0_{2^{t-1}-1}))$$
for $i=0, 1, \ldots, 2^{t-1}-1$ taken in order. 

The minimal polynomial for $R(t)$ has degree $2t$ and is \cite[Theorem 3.4]{us}
\begin{equation} \label{eq:26}
x^{2t} - 2^t.
\end{equation}
We let $\delta_1, \ldots, \delta_{2t}$ denote the roots of \Cref{eq:26}. These roots are obviously distinct.  The characteristic polynomial, say $c_t(x),$ for $R(t)$ has degree $2^t$ and has the same roots $\delta_i,$ but in general some of these roots will occur multiple times. Thus the multiset of roots of $c_t(x),$ counted with multiplicities, will have size $2^t$ but only $2t$ distinct elements.

\begin{theorem} \label{th.char}
The recursion for the weights of $(0,t)_n = h_{t,n},$ extended backwards from $n = 2t+1$ to $n=1,$ generates a sequence $w_1, w_2, \ldots$ (with $w_i = wt(h_{t,i})$ for 
$i \geq 2t+1$) such that
\begin{equation} \label{eq:27}
w_n = 2^{n-1} - \frac{1}{2}(\delta_1^n + \ldots + \delta_{2^t}^n),~n = 1, 2, \ldots.
\end{equation}
Here $\delta_1, \ldots, \delta_{2^t}$ is the list of the $2^t$ roots of the characteristic polynomial $c_t(x)$ for the matrix $R(t),$ with the distinct roots 
$\delta_1, \ldots, \delta_{2t}$ of the minimal polynomial $m_t(x) =x^{2t}-2^t$ for $R(t)$ listed first. The remaining roots are various duplicates of the first $2t$ roots.
\end{theorem}

For monomial functions $(0,t)_n$ we have the following consequence of \cite[Theorem 8]{kph}. Recall that the {\it M\"obius function} \cite[\S 16.3]{hw} is defined by
\begin{equation*}
  \mu(n)=
  \begin{cases}
    1&\text{ if }n=1\\
    (-1)^k &\text{ if }n\text{ is a product of }k\text{ distinct prime factors}\\
    0&\text{ otherwise}
  \end{cases}
\end{equation*}

\begin{theorem}\label{th.deltas}
  Let $t=2^{\nu}m$ for an odd number $m$ and $\nu\ge 0$. The weight $w_n$ of $(0,t)_n$ is expressible as
  \begin{equation*}
    w_n = 2^{n-1} - \frac{1}{2}(\delta_1^n + \ldots + \delta_{2^t}^n),~n = 1, 2, \ldots
  \end{equation*}
  where the multiset $(\delta_i)$ is the union $\sqrt 2\Delta$ of the multisets $\sqrt 2\Delta_d$ indexed by divisors $d|m$, where $\Delta_d$ is the collection of $2^{\nu+1}d^{th}$ roots of unity, each with multiplicity.
  \begin{equation}\label{eq:31}
    \frac
    { \sum_{d' | d} \mu\left(\frac{d}{d'}\right) 2^{2^{\nu}d'} }
    { 2^{\nu+1}d }
  \end{equation}
  where $\mu$ is the M\"obius function. 
\end{theorem}

Implicit in the statement of \Cref{th.deltas} is the remark that the numerator of \Cref{eq:31} is divisible by its denominator. Since every summand $\pm 2^{2^{\nu}d'}$ of the numerator is a multiple of $2^{\nu+1}$, so is the numerator as a whole. On the other hand, divisibility by $d$ follows from

\begin{lemma}
  Let $q,d>1$ be coprime positive integers. Then,
  \begin{equation*}
    D(d,q):=\sum_{d' | d} \mu\left(\frac{d}{d'}\right) q^{d'}
  \end{equation*}
  is divisible by $d$. 
\end{lemma}
\begin{proof}
  If $q$ is a prime power then $\frac{D(d,q)}{d}$ is known to be the number of monic irreducible degree-$d$ polynomials over the field $GF(q)$ with $q$ elements \cite[\S 4.13, Corollary 2]{jac-ba1}. In general, since $q$ and $d$ are assumed coprime Dirichlet's theorem on primes in arithmetic progressions (\cite[\S VI.4, Theorem 2]{se-arith}) ensures that there is some prime congruent to $q$ modulo $d$, reducing the problem to the prime-$q$ case.
\end{proof}

We will also need the following remark. 

\begin{lemma}\label{le.same}
 Let $\{\gamma_i\}$ and $\{\delta_j\}$ be two finite sets of complex numbers and $s_i,t_j\in \bC$. If
 \begin{equation}\label{eq:25}
   \sum_i s_i \gamma_i^n = \sum_j t_j \delta_j^n
 \end{equation}
 for all non-negative integers $n$ then the sets $\{\gamma_i\}$ and $\{\delta_j\}$ and, having identified their respective elements, the corresponding coefficients $s_i$ and $t_j$ also coincide. 
\end{lemma}
\begin{proof}
Suppose the conclusion does {\it not} hold.  Rewriting \Cref{eq:25} as
 \begin{equation*}
   \sum_i s_i \gamma_i^n = \sum_j t_j \delta_j^n
 \end{equation*}
 and aggregating the terms where some $\gamma$ equals some $\delta$,  the failure of the conclusion means that we obtain equations
 \begin{equation*}
   \sum_k u_k \eta_k^n=0,\ \forall n
 \end{equation*}
 for some non-empty set of (distinct) $\eta_k$'s (and non-zero $u_k$). But this means that the vector with components $u_k$ is annihilated by the Vandermonde matrix with entries
 \begin{equation*}
   u_{kl}:=\eta_k^{\ell-1}.
 \end{equation*}
 This contradicts the fact that said matrix has non-zero determinant $\prod_{k>k'}(\eta_k-\eta_{k'})$ and is thus invertible.
\end{proof}

\pf{th.deltas}
\begin{th.deltas}
  According to \Cref{le.same}, it will be enough to show that
  \begin{equation*}
    w_n=2^{n-1}-2^{\frac n2-1}\sum_{\delta\in\Delta}\delta^n. 
  \end{equation*}
  Equivalently, by \cite[Theorem 8]{kph} this amounts to
  \begin{equation}\label{eq:32}
    \sum_{\delta\in \Delta}\delta^n=
    \begin{cases}
      2^{gcd(n,t)}&\text{ if }\frac{n}{gcd(n,t)}\text{ is even}\\
      0&\text{ otherwise}.
    \end{cases}
  \end{equation}
  The second branch is easily dispatched: $\frac{n}{gcd(n,t)}$ being odd is equivalent to $n$ {\it not} being divisible by $2^{\nu+1}$. Since each $\Delta_d$ consists of the $2^{\nu+1}d^{th}$ roots of unity all with the same multiplicity, we have
  \begin{equation*}
    \sum_{\delta\in \Delta_d}\delta^n=0,\ \forall d|m. 
  \end{equation*}
  It thus remains to treat the case when $2^{\nu+1}$ divides $n$, when the target equality \Cref{eq:32} becomes
  \begin{equation}\label{eq:33}
    \sum_{\delta\in \Delta}\delta^n = 2^{gcd(n,t)} = 2^{2^{\nu}gcd(n,m)}. 
  \end{equation}
  Set $D=gcd(n,m)$ for brevity. All
  \begin{equation*}
    \sum_{\delta\in \Delta_d}\delta^n,\ d\not |\; D
  \end{equation*}
  vanish, so we need only consider divisors $d|D$. Keeping this in mind \Cref{eq:33} reads
  \begin{equation*}
    \sum_{d'|d|D}\mu\left(\frac{d}{d'}\right) 2^{2^{\nu}d'}=2^{2^\nu D}.
  \end{equation*}
  This, however, is nothing but an instance of the M\"obius inversion formula \cite[\S 16.4]{hw}.
\end{th.deltas}

By \Cref{eq:32}, proving \Cref{th.char} amounts to showing that for every $n$, we have
\begin{equation}\label{eq:34}
  \mathrm{tr}\;R(t)^n=
  \begin{cases}
    2^{\frac n2+gcd(n,t)}&\text{ if }\frac{n}{gcd(n,t)}\text{ is even}\\
    0&\text{ otherwise}.
  \end{cases}
\end{equation}
Since by \Cref{eq:26} the eigenvalues of $R(t)$ are $2t^{th}$ roots of unity rescaled by $\sqrt 2$, it is enough to prove \Cref{eq:34} for $1\le n\le 2t-1$. It will thus be useful to describe $R(t)^n$ explicitly. To that end, we follow \cite[\S 3]{us} in denoting by $M(n)$ the $2^n\times 2^n$ matrix
\begin{equation*}
  \begin{pmatrix}
    1&\phantom{-}1\\
    1&-1
  \end{pmatrix}^{\otimes n}.
\end{equation*}

For $1\le n\le t$ we also write $M(n,t)$ for the $2^n\times 2^t$ matrix obtained by inserting $2^{t-n}-1$ zero columns after each original column of $M(n)$. For a matrix $M$ we write $\mu(M)$ (or $\mu M$) for the matrix obtained by rotating the rows of $M$ rightward (this extends the above definition of the cyclic permutation $\mu$ on the individual rows).  With all of this in place, the following is a simple computation achievable inductively by partitioning $R(t)$ into four $2^{t-1}\times 2^{t-1}$ block matrices.

\begin{lemma}\label{le.desc-rt}
  Let $1\le n\le 2t-1$. The power $R(t)^n$ can then be described as follows.
  \begin{enumerate}[(a)]
  \item\label{item:5} If $1\le n\le t$ then
    \begin{equation*}
      R(t)^n =
      \begin{pmatrix}
        M(n,t)\\
        \mu M(n,t)\\
        \vdots\\
        \mu^{2^{t-n}-1} M(n,t)
      \end{pmatrix}.
    \end{equation*}
  \item\label{item:6} On the other hand, if $t\le n\le 2t-1$ then
    \begin{equation*}
      R(t)^n = 2^{n-t} \left(R(t)^{2t-n}\right)^T,
    \end{equation*}
    where the $T$ superscript denotes transposition. \qedhere
  \end{enumerate}
\end{lemma}

In particular, part \labelcref{item:6} of \Cref{le.desc-rt} proves \Cref{eq:34} for $t\le n\le 2t-1$ provided it is known for $1\le n\le t$. Even more robustly, it recovers \Cref{eq:34} for a specific $t\le n\le 2t-1$ provided we know the analogue for the reflection $2t-n$ of $n$ across $t$. In conclusion, it suffices to focus on the range $1\le n\le t$. In turn, in those cases the trace of interest is computable as follows, numbering the rows and columns of all matrices starting at $0$:

\begin{lemma}\label{le.how-tr}
  For $1\le n\le t$ the trace $\mathrm{tr}\; R(t)^n$ is the sum of the following elements of $M(n)$:
  \begin{itemize}
  \item the lower right hand corner $M(n)_{2^n-1,2^n-1}$;
  \item for each $0\le k\le 2^n-2$ the entry with index $2^{t-n}k$ modulo $2^n-1$ in the $k^{th}$ column.   \qedhere
  \end{itemize}
\end{lemma}

Note that $2^{dn}-1$ is divisible by $2^n-1$ for all $d\ge 0$, so in \Cref{le.how-tr} is is enough to replace $2^{t-n}$ with the residue $t(\mathrm{mod}\; n)$. We can now rephrase \Cref{le.how-tr} as follows. 

\begin{lemma}\label{le.how-tr-bis}
  Let $1\le n<t$ and set $a=t(\mathrm{mod}\; n)$ and $b=n-a$. Then, the trace $\mathrm{tr}\; R(t)^n$ is the sum of the entries
  \begin{equation}\label{eq:35}
    (q+r2^a,\ q2^b+r)
  \end{equation}
where $0\le q\le 2^a-1$ and $0\le r\le 2^b-1$.   
  \qedhere
\end{lemma}

We make note of the following ``central symmetry'' property of the Hadamard matrix $M(n)$.

\begin{lemma}\label{le.centr}
  Let $a+b=n$ be positive integers and denote by $(k,\ell)$ the coordinates \Cref{eq:35} of an entry in $M(n)$ for some $0\le q\le 2^a-1$ and $0\le r\le 2^b-1$. Let also
  \begin{align*}
    k'=2^n-1-k &= q'+r'2^a\\
    \ell'=2^n-1-\ell &= q'2^b+r'
  \end{align*}
  be the coordinates of the reflection of $(k,\ell)$ across the center of the matrix $M(n)$, where
  \begin{equation*}
    q'+q=2^a-1,\quad r'+r=2^b-1.
  \end{equation*}
  Then,
  \begin{equation*}
    M(n)_{k',\ell'}=
    \begin{cases}
      \phantom{-}M(n)_{k,\ell}&\text{ if }n\text{ is even}\\
      -M(n)_{k,\ell}&\text{ if }n\text{ is odd}.
    \end{cases}
  \end{equation*}
\end{lemma}
\begin{proof}
  We can prove this by induction on $n$, using the decomposition
  \begin{equation*}
    M(n+1)=
    \begin{pmatrix}
      M(n)&\phantom{-}M(n)\\
      M(n)&-M(n)
    \end{pmatrix}
  \end{equation*}
  and treating the separate possibilities for the placement of $(k,\ell)$ in one of the four quadrants. If, say, $M(n+1)_{k,\ell}$ is in the upper left hand $M(n)$ corner and hence
  \begin{equation*}
    M(n+1)_{k,\ell} = M(n)_{k,\ell}
  \end{equation*}
  then $M(n+1)_{k',\ell'}$ is in the lower right hand $-M(n)$ quadrant and hence is {\it minus} the reflection of $M(n)_{k,\ell}$ across the center of $M(n)$. The inductive hypothesis implies the conclusion.

  The argument is analogous in the other cases, and we leave it to the reader; in fact, there is only {\it one} other case: $(k,\ell)$ and $(k',\ell')$ play symmetric roles, so it is enough to assume $(k,\ell)$ is either in the top left or the top right quadrant.
\end{proof}

We can now tackle the following particular case of \Cref{eq:34}.

\begin{corollary}\label{le.oddn}
$\mathrm{tr}\; R(t)^n=0$ when $n$ is odd and hence \Cref{eq:34} holds in that case. 
\end{corollary}
\begin{proof}
  Indeed, \Cref{le.how-tr-bis,le.centr} show that the trace is a sum of pairs $\pm 1$ of entries of $M(n)$, each pair summing to zero.
\end{proof}

\begin{corollary}\label{cor.rminr}
  The matrix $R(t)$ is conjugate to $-R(t)$. 
\end{corollary}
\begin{proof}
  Since the two matrices are unitary and hence diagonalizable over the complex numbers it is enough to argue that they have the same characteristic polynomial. The coefficients of the latter are algorithmically computable from the traces of the powers of the matrix, so it is enough to show that we have
  \begin{equation*}
    \mathrm{tr}\; R(t)^n = \mathrm{tr}\; (-R(t))^n = (-1)^n\mathrm{tr}\; R(t)^n,\ \forall n.  
  \end{equation*}
The two sides are obviously equal for even $n$, so the conclusion follows from \Cref{le.oddn}, which shows that everything in sight vanishes for odd $n$. 
\end{proof}

Recall the polynomials $\Theta_d$ discussed in \Cref{subse.rts}. 

\begin{lemma}\label{le.prod-theta}
  The characteristic polynomial of $R(t)$ is a product of factors $\Theta_d$ for divisors $d|t$. 
\end{lemma}
\begin{proof}
  We already know that $R(t)$ is annihilated by $x^{2t}-2^t$, so its characteristic polynomial will be a product of irreducible factors of the latter. By \Cref{pr.theta-irred}, the conclusion follows from the fact that any two irreducible factors $P$, $Q$ of the characteristic polynomial related by $Q(x)=P(-x)$ have equal exponents because $R(t)$ is conjugate to $-R(t)$ (\Cref{cor.rminr}).
\end{proof}

The roots of $\Theta_d$ are simply those of $\Phi_d(x^2)$ scaled by $\sqrt 2$. In turn, since $d$ divides $t$, the sum of $n^{th}$ powers of the roots of $\Phi_d(x^2)$ equals the sum of $gcd(n,2t)^{th}$ powers. In conclusion:

\begin{lemma}\label{le.div2t}
  It suffices to prove \Cref{eq:34} for $n|2t$. \qedhere
\end{lemma}

There are thus two cases: $n$ divides $t$ or it doesn't, but $\frac n2$ does. The easy half is

\begin{proposition}\label{pr.ndivt}
If $n|t$ then $\mathrm{tr}\; R(t)^n=0$ and hence \Cref{eq:34} holds. 
\end{proposition}
\begin{proof}
  Indeed, in that case $n$ divides $t-n$ and hence
  \begin{equation*}
    2^n-1\ |\ 2^{t-n}-1. 
  \end{equation*}
  \Cref{le.how-tr} then says that $\mathrm{tr}\; R(t)^n$ is precisely the trace of $M(n)$, which is zero, being the $n^{th}$ power of the trace of
  $ \begin{pmatrix}
    1&\phantom{-}1\\
    1&-1
  \end{pmatrix}.
  $
\end{proof}

As for the case $n\not |\;\; t$, we then have $gcd(n,t)=\frac n2$ (since at any rate we are assuming that $n$ divides $2t$) and hence the desired conclusion \Cref{eq:34} reads
\begin{equation*}
  \mathrm{tr}\;R(t)^n=2^n. 
\end{equation*}
\Cref{le.how-tr} then equates this to proving

\begin{lemma}\label{le.all1}
Let $0\le n\le t$ be a divisor of $2t$ but not $t$. Then, for $0\le k\le 2^n-2$, the entry with index $2^{t-n}k$ modulo $2^n-1$ in the $k^{th}$ column of $M(n)$ is $1$.   
\end{lemma}
\begin{proof}
  The hypothesis ensures that $t-n$ is of the form $(2s+1)\frac n2$ for some $s\ge 0$, and hence
  \begin{equation*}
    2^{t-n}k = 2^{sn}2^{\frac n2}k = 2^{\frac n2}k\quad\text{modulo}\quad 2^n-1
  \end{equation*}
  because $2^n-1$ divides $2^{sn}-1$. In short, it will be enough to assume that $t=\frac {3n}2$ thus substituting $2^{\frac n2}$ for $2^{t-n}$ in the statement. 

  We can now partition $M(n)$ into blocks $M_{ij}$ of size  $2^{\frac n2}\times 2^{\frac n2}$ for $0\le i,j\le 2^{\frac n2}-1$, each a copy of either $M(\frac n2)$ or $-M(\frac n2)$. The entries of interest in the first $2^{\frac n2}$ columns are
  \begin{itemize}
  \item the $0^{th}$ entry in the $0^{th}$ row of $M_{00}$;
  \item the $1^{st}$ entry in the $0^{th}$ row of $M_{10}$;
  \item $\cdots$;
  \item entry $(M_{2^{\frac n2}-1,0})_{0,2^{\frac n2}-1}$.
  \end{itemize}
  The pattern recurs: including the bottom right corner of $M(n)$, the entries we are after are precisely those of the form $(M_{ij})_{ji}$. That these are all $1$ follows from the recursive construction of $M(n)$ giving
  \begin{equation*}
    M(n)=M\left(\frac n2\right)\otimes M\left(\frac n2\right)
  \end{equation*}
  together with the fact that $M$ is symmetric. 
\end{proof}

\pf{th.char}
\begin{th.char}
  As discussed above, the result amounts to \Cref{eq:34}. In turn, the latter is taken care of by \Cref{le.div2t,pr.ndivt} and \Cref{le.all1}.
\end{th.char}

\Cref{th.char} states that the formula which gives the weights $w_n$ for the MRS function $(0,t)_n$ in terms of powers of the roots of the characteristic polynomial has simple coefficients which are all $\pm \frac{1}{2}.$  We say that the recursion for the weights of $(0,t)_n$ has {\it easy coefficients}.  We believe this remains true for {\it any} quadratic RS function and so state the following conjecture:

\begin{conjecture} [\bf Easy Coefficients Conjecture]
  The recursion for the weights of any rotation symmetric function has easy coefficients, attached to a multiset of algebraic integers.  At least in the quadratic case, these algebraic integers are the roots of the characteristic polynomial of a matrix computable by the method of \cite{csk-rec2}.
\end{conjecture}

The first sentence of the Easy Coefficients Conjecture is proved by \Cref{cor.sum-pows}.  The second sentence for MRS quadratic functions is proved by \Cref{th.char}.
Since there is no nice formula like  \Cref{eq:26} for the minimal polynomial for the square matrix $R$ (generalization of $R(t)$ for  $(0,t)_n$) corresponding to a general quadratic RS function, the method of proof of \Cref{th.char} does not apply. However, we are confident that the result is true for general quadratic RS functions. In fact, many computations suggest that the second sentence of the Easy Coefficients Conjecture is true for{ \it all} RS functions, of any degree.
 
Given a recursion relation of order $r$ for the weights $w_n$ of any RS Boolean function in $n$ variables, the standard way to compute the weights is to compute the needed initial $r$ weights and then use the recursion to find further desired weights.  It is well known that the runtime to find $w_n$ is $O(n2^n)$ (the extra $n$ comes from the operations needed to compute each entry in the truth table), so finding the initial conditions in this way is an exponential computation. Given a function for which the Easy Coefficients Conjecture is true, a much quicker way to find the initial weights is to compute the roots of the characteristic polynomial and then use the analog of \Cref{eq:27}.  The problem of computing all of the roots of a polynomial with integer coefficients has been studied for a long time.  The runtime for doing that is known to be $O(n(log^k~n))$ for some small integer $k,$ so an exponential computation in $n$ has been replaced by one that is nearly {\it linear} in $n.$




\bibliographystyle{plain}

\Addresses

\end{document}